\pdfoutput=1

\documentclass{article}

\usepackage{fancyhdr}
\usepackage[utf8]{inputenc}
\usepackage[american]{babel}
\usepackage{amssymb}
\usepackage{amsmath}
\usepackage{amsthm}
\usepackage{listings}
\usepackage{graphicx}
\usepackage{color}
\usepackage{subfig}
\usepackage{xspace}
\usepackage{enumerate}
\usepackage[inline]{enumitem}
\usepackage{xargs}
\usepackage{mathtools} 
\usepackage{nicefrac}
\usepackage{dsfont}
\usepackage{algpseudocode}
\usepackage{todonotes}
\usepackage{authblk}
\usepackage{hyperref}
\usepackage{tikz}
\usetikzlibrary{arrows}
\usetikzlibrary{patterns} 
\usetikzlibrary{decorations.pathreplacing}
\usetikzlibrary{arrows.meta}

\usepackage[capitalise,nameinlink,noabbrev]{cleveref}
\crefname{ineq}{Inequality}{Inequalities}
\creflabelformat{ineq}{#2{\upshape(#1)}#3}
\crefname{enumi}{}{}
\crefname{step}{Step}{Steps}
\creflabelformat{step}{#2#1#3.}
\crefname{lemma}{Lemma}{Lemmas}
\crefname{corollary}{Corollary}{Corollaries}
\crefname{proposition}{Proposition}{Propositions}
\crefname{property}{Property}{Properties}

\newtheorem{theorem}{Theorem}[section]
\newtheorem{lemma}[theorem]{Lemma}

\newtheorem{corollary}[theorem]{Corollary}

\newtheorem{remark}[theorem]{Remark}
\newcommand*{\OPT}{\ensuremath{\textsc{Opt}}\xspace}

\newcommand{\eps}{\varepsilon}

\newcommand{\ZZ}{\mathbb{Z}}

\newcommand{\machs}{\mathcal{M}}
\newcommand{\jobs}{\mathcal{J}}

\newcommand{\states}{\mathcal{S}}
\newcommand{\ijobs}{\mathcal{I}}
\newcommand{\fjobs}{\mathcal{F}}
\newcommand{\corejobs}{\bar{J}}
\newcommand{\fringejobs}{\tilde{J}}

\DeclarePairedDelimiter\floor{\lfloor}{\rfloor}
\DeclarePairedDelimiter\ceil{\lceil}{\rceil}

\DeclarePairedDelimiter\set{\lbrace}{\rbrace}
\DeclarePairedDelimiterX\sett[2]{\lbrace}{\rbrace}{ #1 \,\delimsize| \,\mathopen{} #2 }

\begin{document}

\title{Scheduling on (Un-)Related Machines\\ with Setup Times\thanks{This work was partially supported by the German Research Foundation (DFG)
within the Collaborative Research Centre ``On-The-Fly Computing'' (SFB 901) and by DFG project JA 612/20-1}
}
\author[1]{Klaus Jansen}
\author[1]{Marten Maack}
\author[2]{Alexander M\"acker}
\affil[1]{Department of Computer Science, University of Kiel, Kiel, Germany}
\affil[2]{Heinz Nixdorf Institute and Computer Science Department, Paderborn University, Paderborn, Germany}
\affil[1]{\textit {\{kj,mmaa\}@informatik.uni-kiel.de}}
\affil[2]{\textit {alexander.maecker@uni-paderborn.de}}

\date{}
\maketitle
\abstract{We consider a natural generalization of scheduling $n$ jobs on $m$ parallel machines so as to minimize the makespan.
In our extension the set of jobs is partitioned into several classes and a machine requires a setup whenever it switches from processing jobs of one class to jobs of a different class.
During such a setup, a machine cannot process jobs and the duration of a setup may depend on the machine as well as the class of the job to be processed next.

For this problem, we study approximation algorithms for non-identical machines.
We develop a polynomial-time approximation scheme for uniformly related machines.
For unrelated machines we obtain an $O(\log n + \log m)$-approximation, which we show to be optimal (up to constant factors) unless $NP \subset RP$.
We also identify two special cases that admit constant factor approximations.
}

\section{Introduction}
\label{sec:intro}
We consider a problem that is a natural generalization of the classical parallel machine scheduling problem:
We are given a set of $n$ jobs as well as $m$ parallel (identical, uniformly related or unrelated) machines and the goal is to find an assignment of jobs to machines so as to minimize the makespan.
Our generalization assumes the set of jobs to be partitioned into $K$ classes and a machine needs to perform a setup whenever it switches from processing a job of one class to a job of a different class.
Thereby, the length of the setup may depend on the machine as well as the class of the job to be processed, but does not depend on the class previously processed on the machine.
Such an explicit modeling of (sequence-(in)dependent) setup times has several applications/motivations: they occur in production systems, for example, as changeover times, times for cleaning activities or for preparations such as the calibration of tools; or in computer systems, for example, due to the transfer of data required to be present at the executing machine \cite{ali1,ali2,ali3}.

In the past, approximation algorithms for this problem have been designed for the case of identical machines \cite{wads15,europar16,martenConfigIP}, however, not much is known about the non-identical case.
The goal of this paper is to advance the understanding of the problem in case the machines have different capabilities, which we capture by modeling them as uniformly related or unrelated machines.
This seems to be an important topic as it is a natural special case of the following problem, which is quite present in the literature on heuristics and exact algorithms (cf.\ \cite{ali1,ali2,ali3}), but lacks (to the best of our knowledge) theoretical investigations with provable performance guarantees: 
Jobs need to be processed on parallel unrelated machines and each job has a setup time that might depend on the machine as well as the preceding job.
Note that in this paper we require the setup times to have a certain regular structure in the sense that it is $0$ for a job $j$ if $j$ is preceded by a job of the same class and otherwise it only depends on $j$'s class and the machine.

\subsection{Model \& Notation}
\label{sec:model}
We consider a scheduling problem that generalizes the classical problem of minimizing the makespan on parallel machines:
In our model, we are given a set $\mathcal{J}$ of $n \coloneqq |\mathcal{J}|$ jobs as well as a set $\mathcal{M}$ of $m \coloneqq |\mathcal{M}|$ parallel machines.
Each job $j \in \mathcal{J}$ comes with a processing time (size) $p_{ij} \in \mathbb{N}_{\geq 0}$ for each $i \in \mathcal{M}$.
Additionally, the set $\mathcal{J}$ of jobs is partitioned into $K$ classes $\mathcal{K}$.
Each job $j$ belongs to exactly one class $k_j \in \mathcal{K}$ and with each class $k \in \mathcal{K}$ and machine $i \in \mathcal{M}$ a setup time $s_{ik}\in \mathbb{N}_{\geq 0}$ is associated.
The goal is to compute a non-preemptive schedule in which each job is processed on one machine and each machine processes at most one job at a time and which minimizes the makespan:
A schedule is given by a mapping $\sigma: \mathcal{J} \to \mathcal{M}$ and the goal is to minimize (over all possible $\sigma$) the \emph{makespan} $\max_{i \in \mathcal{M}} L_i$ given by the maximum \emph{load} of the machines $L_i \coloneqq \sum_{j\in\sigma^{-1}(i)} p_{ij} + \sum_{k \in \{k_j : j \in \sigma^{-1}(i)\} }s_{ik}$.
Intuitively, one can think of the load of a machine as the processing it has to do according to the jobs assigned to it plus the setups it has to pay for classes of which it does process jobs.
This reflects problems where a machine $i$ processes all jobs belonging to the same class in a batch (a contiguous time interval) and before switching from processing jobs of a class $k'$ to jobs of class $k$ it has to perform a setup taking $s_{ik}$ time.
For simplicity of notation, for a fixed problem instance and an algorithm $\mathcal{A}$, we denote the makespan of the schedule computed by $\mathcal{A}$ as $|\mathcal{A}|$.

In the most general model for parallel machines, the \emph{unrelated} machines case, there are no restrictions on the processing times $p_{ij}$, which therefore can be completely arbitrary.
In case of \emph{uniformly related} machines, each machine $i$ has a fixed speed $v_i$ and the processing time $p_{ij}$ only depends on the job $j$ and the speed of machine $i$ and is given by $p_{ij} = \frac{p_j}{v_i}$.
Finally, we consider the \emph{restricted assignment} problem, where each job $j$ has a set $M_j$ of eligible machines (on which it can be processed) and the processing time is the same on all of them, that is, $p_{ij} = p_j$ for all $i \in M_j$ and $p_{ij}=\infty$ otherwise.

For each of these variants we assume that the setup times behave similar to the jobs, that is, in the unrelated case we have arbitrary setup times $s_{ik}$ depending on the machine $i$ and the class $k$; in the uniform case we have, $s_{ik} = \frac{s_k}{v_i}$; and in the restricted assignment case, we have $s_{ik}\in\set{s_k,\infty}$.
This model seems sensible, if we assume that the different behavior is due to qualitative differences between the machines, like suggested by the names of the problems.  

\subsubsection{Further Notions}
\label{sec:notions}
A polynomial time (approximation) algorithm $\mathcal{A}$ is called to have an \emph{approximation factor} $\alpha$ if, on any instance, 
$|\mathcal{A}| \leq \alpha |\OPT|$ holds, where $|\OPT|$ denotes the optimal makespan.
In case $\mathcal{A}$ is a randomized algorithm, we require that $\mathbb{E}[|\mathcal{A}|] \leq \alpha |\OPT|$, where the expectation is taken with respect to the random choices of $\mathcal{A}$.
An approximation algorithm is called a \emph{polynomial time approximation scheme} (PTAS) if, for any $\varepsilon >0$, it computes a $(1+\varepsilon)$-approximation in time polynomial in the input size and (potentially) exponential in $\frac{1}{\varepsilon}$.

Our approximation algorithms almost all follow the idea of the \emph{dual approximation framework} introduced by Hochbaum and Shmoys in \cite{dualApprox}.
Instead of directly optimizing the makespan, we assume that we are given a bound $T$ on the makespan and we are looking for an algorithm that computes a schedule with makespan (at most) $\alpha T$ or correctly decides that no schedule with makespan $T$ exists.
Employing this idea, it is easy to see that using binary search started on an interval $I \ni  |\OPT|$ that contains the optimal makespan, finally provides an approximation algorithm with approximation factor $\alpha$.

\subsection{Related Work}
\label{sec:relatedWork}

\emph{Uniformly Related Machines. }
As already discussed, our model can be viewed as a generalization of classical parallel machine models without setup times (where all setup times are $0$).
For these models, it is known for a long time due to the work of Hochbaum and Shmoys \cite{ptasUniform} that a PTAS can solve the problem of uniformly related machines arbitrarily close to optimal.
More recently, in \cite{eptasUniform} Jansen even shows that the running time can be further improved by coming up with an EPTAS, a PTAS with running time of the form $f(1/\varepsilon) \times \text{poly}(|I|)$, where $f$ is some computable function and $|I|$ the input size.

\emph{Unrelated Machines and Restricted Assignment. } The case of unrelated machines significantly differs from the uniform case due to an inapproximability result of $3/2$ (unless P$=$NP) as proven by Lenstra, Shmoys and Tardos in \cite{lenstraUnrelated}. 
On the positive side, there are algorithms that provide $2$-approximations based on rounding fractional solutions to a linear programming formulation of the problem.
A purely combinatorial approach with the same approximation factor is also known \cite{martin}.
For special cases of the restricted assignment problem stronger results are known, e.g.,  Ebenlendr et al.\ \cite{graphBalancing} show that the lower bound of $3/2$ even holds for the more restrictive case where $|M(j)| \leq 2$ for all $j$, and design a $1.75$-approximation algorithm for this case.
For the general restricted assignment case, Svensson \cite{RAestimation} provides an algorithm for estimating the optimal makespan within a factor of $\frac{33}{17}$.
Jansen and Rohwedder \cite{RAbetterestimation} improve this to $\frac{11}{6}$ and also \cite{RAQuasipoly} give an algorithm with quasipolynomial running time and approximation ratio $\frac{11}{6} + \eps$.

\emph{Setup Times. } Scheduling with an explicit modeling of setup times has a long history, particularly within the community of operations research.
The vast majority of work there studies hardness results, heuristics and exact algorithms, which are evaluated through simulations, but without formal performance guarantees.
The interested reader is referred to the exhaustive surveys on these topics by Allahverdi et al.\ \cite{ali1,ali2,ali3}.
In contrast, literature in the domain of approximation algorithms with proven bounds on the performance is much more scarce.
Schuurman and Woeginger \cite{preemptiveSetups}, consider a model where jobs are to be processed on identical machines in a preemptive way so as to minimize the makespan.
Whenever a machine switches from processing one job to a different job, a setup time is necessary.
Schuurman and Woeginger design a PTAS for the case of job-indepedent setup times and a $4/3$-approximation for the case of job-dependent setup times.
In \cite{correa15}, Correa et al.\ consider a similar model where jobs can not only be preempted but be split arbitrarily (thus, job parts can also be processed simultaneously on different machines). 
They design a $(1+\phi)$-approximation, where $\phi \approx 1.618$ is the golden ratio, for the case of unrelated machines as well as an inapproximability result of $\frac{e}{e-1}$.
The model with classes and (class-independent) setups was first considered by M\"acker et al.\ for identical machines in \cite{wads15}, where constant factor approximations are presented.
In \cite{europar16}, Jansen and Land improve upon these results by providing a PTAS (even) for the case of class-dependent setup times.
This result has been further improved in \cite{martenConfigIP} to an EPTAS.
The same work \cite{martenConfigIP} also improves on the result from \cite{preemptiveSetups} (mentioned above) by giving an EPTAS for the respective problem and obtains an EPTAS for the identical machines case of the model given in \cite{correa15} (discussed above).

Our model with classes and setup times has also been considered for objective functions other than makespan by Divakaran and Saks for a single machine.
In \cite{saksCompletionTime}, they give a $2$-approximation for the weighted completion time objective and an algorithm achieving a maximum lateness that is at most the maximum lateness of an optimal solution for a machine running at half the speed.
In \cite{saksFlowtime}, they design and analyze an (online) algorithm  having a constant approximation factor for minimizing the maximum flow time.
In \cite{splittingVsSetups}, Correa et al.\ study the objective of minimizing the weighted completion time in the setting where jobs can be split arbitrarily and each part requires a setup before being processed.
They propose constant factor approximations for identical and unrelated machines.

\subsection{Our Results}
\label{sec:results}
In \cref{sec:uniform}, we present the first PTAS for scheduling on uniformly related machines with setup times. 
Roughly speaking, our main technical contribution, is to simplify the problem, such that for each setup class the setup times can be ignored on all machines but those whose speeds belong to some bounded (class dependent) interval of machine speeds.
In \cref{sec:unrelated} we study the case of unrelated machines and start with a randomized rounding based algorithm to compute $O(\log n + \log m)$-approximations in \cref{sec:unrelatedApprox}.
We prove that this bound is (asymptotically) tight (unless $NP=RP$) by providing a randomized reduction from the \textsc{SetCover} problem in \cref{sec:unrelatedReduction}.
We conclude in \cref{sec:specialUnrelated} with identifying two special cases of unrelated machines that admit constant factor approximations by showing how a rounding technique from \cite{correa15} can be employed to approximate these cases.

\section{Uniformly Related Machines}
\label{sec:uniform}
In this section, we develop a PTAS for uniformly related machines based on a dual approximation.
To bootstrap the dual approximation framework, that is, to determine a (small) interval containing $|\textsc{Opt}|$, we could, with a very efficient runtime (dominated by time for sorting) of $O(n \log n)$, compute a constant factor approximation based on the standard LPT-rule as follows:
Let $J_s^k = \{j \in \mathcal{J} : k_j = k, p_j < s_k\}$ be the set of jobs of class $k$ being smaller than the setup time of $k$.
Replace the jobs in $J_s^k$ by $\lceil \sum_{j \in J_s^k}p_j/s_k\rceil$ many (placeholder) jobs of class $k$, each with a size of $s_k$.
Then apply the standard LPT-rule ignoring any classes and setups (that is, sort all jobs by non-increasing size and add one after the other to the machine where it finishes first); and finally, add all required setups to the LPT-schedule and replace the placeholder by the actual jobs.
As LPT provides $(1+\frac{1}{\sqrt{3}})$-approximations for scheduling on uniformly related machines \cite{lpt}, a straightforward reasoning shows this approach to provide $3(1+\frac{1}{\sqrt{3}}) \approx 4.74$-approximations.

\begin{lemma}
Using the LPT-rule as described above, yields an approximation factor of $3(1+\frac{1}{\sqrt{3}}) \approx 4.74$.
\end{lemma}
\begin{proof}
Consider an optimal schedule and let $S_i$ be the set of classes for which there is a setup on machine $i$.
Then, there is a schedule with load at most $|\textsc{Opt}|+\sum_{k \in S_i} s_{ik}$ on machine $i$ after replacing the small jobs by placeholder jobs.
Also, when ignoring any setups, this load is decreased to at most $|\textsc{Opt}|$.
Therefore, using LPT, we find a schedule with makespan at most $(1+\frac{1}{\sqrt{3}})|\textsc{Opt}|$.
Now let $C_i$ be the set of classes of which jobs are scheduled on machine $i$ in the LPT schedule.
Replacing the placeholder by actual jobs can increase the makespan by at most $\sum_{k \in C_i} s_{ik}$ and adding the required setups can increase it by the same amount.
Since $\sum_{k \in C_i} s_{ik} \leq (1+\frac{1}{\sqrt{3}})|\textsc{Opt}|$, the lemma follows.
\end{proof}

\subsection{PTAS}

The roadmap for the PTAS is as follows:
\begin{enumerate}
\item Simplify the instance.
\item Find a relaxed schedule for the simplified instance via dynamic programming, or conclude correctly that no schedule with makespan $T$ for the original instance exists.
\item Construct a regular schedule for the simplified instance using the relaxed schedule and a greedy procedure.
\item Construct a schedule for the original instance using the one for the simplified instance.
\end{enumerate}

Concerning the second and third step, first note that the makespan guess $T$, given by the dual approximation framework, enables a packing perspective on the problem:
On machine $i$ there is an amount of $Tv_i$ free space and the jobs and setup times have to be placed into this free space.
Now, a job or setup time may be big or small relative to this free space, say bigger or smaller than $\eps T v_i$.
In the latter case, $i$ can receive one additional job or setup time in a PTAS, or several for another threshold parameter than $\eps$.
Hence, we have to be cautious when placing big objects but can treat small objects with less care. 
Roughly speaking, in a relaxed schedule some jobs and setups are fractionally placed on machines for which they are small, and for jobs that are big relative to the setup time of their class, the setup is ignored.

For the dynamic program, we define intervals of machine speeds, called groups, and the groups are considered one after another ordered by speeds and starting with the slowest.
In each interval, the speeds differ at most by a constant factor.
This enables us to reuse ideas for the identical machine case developed in \cite{europar16} for the single groups.
However, there has to be some information passed up from one group to the next, and this has to be properly bounded, in order to bound the running time of the dynamic program.
While, we can use some standard ideas for classical makespan minimization on uniformly related machines (without setup times), e.g. from \cite{ptasUniform}, there are problems arising from the setup classes.
Mainly, we have to avoid passing on class information between the groups.
As a crucial step to overcome this problem, we show that for each group there is only a bounded interval of machine speeds for which we have to properly place the setup times.
In the algorithm, we define the groups wide enough and with overlap such that for each class there is a group containing the whole interval relevant for this class.
When going from one group to the next, we therefore do not have to pass on class information of jobs that have not been scheduled yet.
This, together with proper simplification steps enables us to properly bound the running time of the dynamic program.

In the following, we describe the PTAS in detail, starting with the simplification steps, followed by some definitions and observations that lead to the definition of a relaxed schedule, and lastly, we present the dynamic program.

Throughout this section $\eps > 0$ denotes the accuracy parameter of the PTAS with $1/\eps \in \ZZ_{\geq 2}$; and $\log(\cdot)$ the logarithm with basis 2. 
Furthermore, for a job $j$ or a setup class $k$, we call the values $p_j$ and $s_k$ the job or setup size respectively, in distinction from their processing time $p_{ij} = p_j/v_i$ or setup time $s_k/v_i$ on a given machine $i$.

\paragraph{Simplification Steps.}

We perform a series of simplification steps:
First, we establish minimum sizes of the occurring speeds, job and setup sizes; next, we ensure that the job sizes of a class are not much smaller than its setup size; and lastly, we round the speeds, job and setup sizes. 
Most of the used techniques, like geometric rounding or the replacement of small objects with placeholders with a minimum size, can be considered folklore in the design of approximation algorithms for scheduling problems.
Similar arguments can be found, e.g., in \cite{europar16}, \cite{ptasUniform}, \cite{eptasUniform} or \cite{epstein2004approximation}.

Let $I$ be the original instance and $v_{\max} = \max\sett{v_i}{i\in\machs}$.
We remove all machines with speeds smaller than $\eps v_{\max}/m$ and denote the smallest remaining speed after this step by $v_{\min}$.
Furthermore, we increase all job and setup sizes that are smaller than $\eps v_{\min}T/(n+K)$ to this value, and call the resulting instance $I_1$.
By scaling, we assume $v_{\min}T = 1$ in the following.
\begin{lemma}
If there is a schedule with makespan $T$ for $I$, there is also a schedule with makespan $(1+\eps)^2T$ for $I_1$; and if there is a schedule with makespan $T'$ for $I_1$, there is also a schedule with makespan $T'$ for $I$.
\end{lemma}
\begin{proof}
Given a schedule for $I$, the summed up load on machines missing in $I_1$ is upper bounded by $\eps v_{\max}T$ and we can place it on a fastest machine.
Furthermore, increasing the setup and processing times can increase the load on any machine by at most $\eps v_{\min}T$.
\end{proof}

The next step is to make sure that jobs are not much smaller than the setup size of their class.
Let $I_2$ be the instance we get by replacing for each class $k$ the jobs with size smaller than $\eps s_k$ with placeholders, that is, we remove the jobs from $\jobs'_{k} = \sett{j\in\jobs_k}{p_j\leq \eps s_k}$ and introduce $\ceil{(\sum_{j\in\jobs'_{k}}p_j)/(\eps s_k)}$ many jobs of size $\eps s_k$ belonging to class $k$.
\begin{lemma}
If there is a schedule with makespan $T'$ for $I_1$, there is one with makespan $(1+\eps)T'$ for $I_2$; and if there is a schedule with makespan $T'$ for $I_2$, there is one with makespan $(1+\eps)T'$ for $I_2$.
\end{lemma}
\begin{proof}
Given a schedule for one of the instances, we can greedily replace jobs with the respective placeholders and vice-versa, over-packing with at most one object per class and machine.
Thereby the overall load on each machine due to a class scheduled on the machine is increased at most by a factor of $(1+\eps)$.
\end{proof}

Next, we perform rounding steps for the job and setup sizes, as well as the machine speeds:
For each job or setup size $t$, let $e(t) = \floor{\log t}$.
We round $t$ to $2^{e(t)} + k\eps 2^{e(t)}$ with $k = \ceil{(t-2^{e(t)})/(\eps 2^{e(t)})}$.
This rounding approach is due to Gálvez et al. \cite{Verschae16}.
Furthermore, we perform geometric rounding for machine speed, that is, each machine speed $v$ is rounded to $(1 + \eps)^{k'}v_{\min}$, with $k' = \floor{\log_{1+\eps}(v_i/v_{\min})}$.
We call the rounded instance $I_3$.
\begin{lemma}
If there is a schedule with makespan $T'$ for $I_2$, there is also a schedule with makespan $(1+\eps)^2T'$ for $I_3$; and if there is a schedule with makespan $T'$ for $I_3$, there is also one for $I_2$.
\end{lemma}
\begin{proof}
Each job and setup size is increased at most by a factor of $(1+\eps)$ by the rounding and each machine speed is decreased at most by a factor of $(1+\eps)$.
\end{proof}

Hence, if there is a schedule with makespan at most $T$ for $I$ there is also a schedule with makespan at most $T_1$ for $I_3$ with $T_1 = (1+\eps)^5T = (1+O(\eps))T$.
Furthermore, if we should find a schedule with makespan $T_2$ for $I_4$ with $T_1\leq T_2 = (1+O(\eps))T$, we can transform it back into a schedule for the original instance with makespan at most $T_3 = (1+\eps)T_2 = (1+O(\eps))T$.

For the sake of simplicity, we assume in the following that the instance $I$ is already simplified and the makespan bound $T$ was properly increased.

\paragraph{Preliminaries.}

We define two threshold parameters $\delta = \eps^2$ and $\gamma = \eps^3$.
For each class $k$ the \emph{core jobs} belonging to that class are the ones with a job size $p$, such that $\eps s_k \leq p < s_k /\delta$.
Bigger jobs are called \emph{fringe jobs}.
The set of core or fringe jobs of class $k$ is denoted by $\corejobs_k$ and $\fringejobs_k$ respectively.
The \emph{core machines} $i$ of class $k$, are the ones with $s_k \leq Tv_i < s_k /\gamma$ and faster machines are called \emph{fringe machines}.
\begin{remark}
For each class $k$ and each job $j$ that belongs to $k$, $j$ is either a core or a fringe job and has to be scheduled either on a core or a fringe machine of $k$. 
\end{remark}

A job size $p$ is called \emph{small} for a speed $v$, if $p < \eps vT$; \emph{big}, if $\eps vT \leq p \leq vT$; and \emph{huge}, if $p > vT$.
We use these terms for jobs and machines as well, e.g., we call a job $j$ small for machine $i$, if $p_j < \eps v_i T$.
Since $\gamma/\delta = \eps$ holds, we have:
\begin{remark}
The core jobs of class $k$ are small on fringe machines of $k$. 
\end{remark}

Next, we define speed groups (see Fig. \ref{fig:groups}).
For each $g\in \ZZ$, we set $\check{v}_g = v_{\min}/\gamma^{g-1}$ and $\hat{v}_g = v_{\min}/\gamma^{g+1}$.
Group $g$ is given by the interval $[\check{v}_g, \hat{v}_g)$.
Note that the groups are overlapping with each speed occurring in exactly two groups.
A machine $i$ belongs to group $g$, if $v_i\in[\check{v}_g, \hat{v}_g)$, and we denote the set of machines belonging to $g$ by $M_g$ and the set of corresponding speeds by $V_g$, i.e., $V_g = \sett{v_i}{i\in M_g}$.
By definition, the smallest group $g$ with $M_g \neq \emptyset$ is group $0$.
Furthermore, let $G$ be the biggest number with this property.
Because of the first simplification step, we have $G\leq m/(3\eps\log(1/\eps)) = O(m/\eps)$.

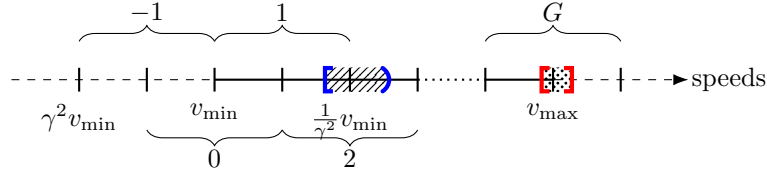
\begin{figure}
	\centering
	\begin{tikzpicture}[scale=9]
	\draw[dashed] (0,0) -- (0.3,0);
	\draw[thick] (0.3,0) -- (0.6,0);
	\draw[dotted, thick] (0.6,0) -- (0.7,0);
	\draw[thick] (0.7,0) -- (0.8,0);
	\draw[-{Latex[length=2mm]},dashed] (0.8,0) -- (1,0);
	\foreach \x/\xtext in {0.1/$\gamma^2 v_{\min}$,0.2/,0.3/$v_{\min}$,0.4/,0.5/$\frac{1}{\gamma^2}v_{\min}$,0.6/,0.7/,0.8/$v_{\max}$,0.9/}
	\draw[thick] (\x,0.5pt) -- (\x,-0.5pt) node[below,yshift=-1pt] {\xtext};
	\draw[decorate,decoration={brace,amplitude=6pt,raise=1pt},yshift=1.5pt] (0.1,0) -- (0.3,0) node [midway,yshift=12pt]{$-1$};
	\draw[decorate,decoration={brace,amplitude=6pt,raise=1pt},yshift=1.5pt] (0.3,0) -- (0.5,0) node [midway,yshift=12pt]{$1$};
	\draw[decorate,decoration={brace,amplitude=6pt,raise=1pt},yshift=1.5pt] (0.7,0) -- (0.9,0) node [midway,yshift=12pt]{$G$};
	\draw[decorate,decoration={brace,amplitude=6pt,raise=1pt,mirror},yshift=-2pt] (0.2,0) -- (0.4,0) node [midway,yshift=-12pt]{$0$};
	\draw[decorate,decoration={brace,amplitude=6pt,raise=1pt,mirror},yshift=-2pt] (0.4,0) -- (0.6,0) node [midway,yshift=-12pt]{$2$};
	
	\fill[pattern = north east lines,rounded corners=1ex] (0.461, -0.12ex) -- (0.559, -.12ex) -- (0.559, .12ex) -- (0.461,.12ex) -- cycle;
	\draw[[-, ultra thick, blue] (0.46,0) -- (0.461,0);
	\draw[-), ultra thick, blue] (0.56,0) -- (0.561,0);
	
	\fill[pattern = crosshatch dots,rounded corners=1ex] (0.781, -0.12ex) -- (0.829, -.12ex) -- (0.829, .12ex) -- (0.781,.12ex) -- cycle;
	\draw[[-, ultra thick, red] (0.78,0) -- (0.781,0);
	\draw[{-]}, ultra thick, red] (0.83,0) -- (0.831,0);
	
	\draw (1.06,0) node {speeds};
	\end{tikzpicture}
	\caption{Machine speeds with logarithmic scale. The braces mark groups; the dashed interval, possible speeds of core machines of some class with core group $2$; and the dotted interval, possible speeds of machines where some job with native group $G$ is big.}
	\label{fig:groups}
\end{figure}

For each job $j$ there are up to three (succeeding) groups containing speeds for which its size is big, and at least one of them contains all such speeds. 
Let $g$ be the smallest group with this property, i.e., $p_j \geq \eps\check{v}_g T$ and $p_j < \hat{v}_g T$.
We call $g$ the \emph{native group} of $j$.
For a group $g$, the fringe jobs with native group $g$ will be of interest in the following and we denote the set of these jobs by $\fringejobs_g$.

Moreover, for each class $k$ there are at most three (succeeding) groups containing possible speeds of core machines of $k$, and there is at least one that contains all of them.
Let $g$ be the smallest group with this property, i.e., $s_k \geq \check{v}_g T$ and $s_k < \hat{v}_g T$.
We say that $g$ is the \emph{core group} of $k$.
Note that $k$ has a core group even if it has no core machines.
\begin{remark}
Let $j$ be a core job of class $k$ and $g$ be the core group of $k$.
There is a speed $v$ in group $g$ such that $p_j$ is big for $v$.
\end{remark}
We have $p_j < s_k/\eps^2$ because $j$ is a core job; and $s_k/\eps^2<\eps\hat{v}_gT$, because $g$ is the core group of $k$.
Hence, $p_j$ is small for $\hat{v}_g$.
Furthermore, we have $p_j \geq \eps s_k \geq \eps \check{v}_g T$ for the same reasons.
Therefore, $p_j$ is big or huge for $\check{v}_g$ and there lies at least one speed in between for which it is big.

\paragraph{Relaxed Schedule.}

In a \emph{relaxed schedule}, the set of jobs is partitioned into integral jobs $\ijobs$ and fractional jobs $\fjobs$, and an assignment $\sigma': \ijobs \rightarrow \machs$ of the integral jobs is given.
For each $j\in\ijobs$ the machine $\sigma'(j)$ belongs to the native group of $j$, if $j$ is a fringe job, and to the core group of $k$, if $j$ is a core job of class $k$.
Setups for fringe jobs are ignored, and hence we define the relaxed load $L'_i$ of machine $i$ to be $\sum_{j\in\sigma'^{-1}(i)}p_j + \sum_{k:\sigma'^{-1}(i)\cap\corejobs_k\neq\emptyset}s_k$.
Intuitively, the fractional jobs are placed fractionally together with some minimum amount of setup in the left-over space on the machines that are faster than the ones in the respective native or core group.
More formally, we say that the relaxed schedule has makespan $T$ if $L'_i\leq T$ for each $i\in\machs$ and the following space condition for the fractional jobs holds.

Let $\fjobs_g$ be the set of fractional fringe jobs with native group $g$, and fractional core jobs of class $k$ with core group $g$; $A_i = \max\{0,Tv_i-L'_i\}$ the remaining free space on machine $i$ with respect to $T$; and $W_g$ the overall load of fractional jobs with native group $g$ together with one setup for each class that 
\begin{enumerate*}
\item has core group $g$,
\item has no fringe job, and
\item has a fractional core job, 
\end{enumerate*}
i.e., $W_g = \sum_{j\in\fjobs_g} p_j + \sum_{k:\fjobs_g\cap\corejobs_k\neq\emptyset, \fringejobs_k=\emptyset} s_k$.
A job $j\in\fjobs_g$ should be placed on a machine that belongs to group $g+2$ or a faster group.
Hence, we set the reduced accumulated fractional load $R_g$ for group $g$ to be $\max\set{0,R_{g-1} + W_{g-2} - \sum_{i\in M_g\setminus M_{g+1}}A_i}$.
The required space condition is $R_{G} = W_{G} = W_{G-1} = 0$.

\begin{lemma}
If there is a schedule with makespan $T$ for a given instance, there is also a relaxed schedule with makespan $T$; and if there is a relaxed schedule with makespan $T$, there is a schedule with makespan $(1+\mathcal{O}(\eps))T$.
\end{lemma}
\begin{proof}
The first claim is easy to see:
For a given schedule $\sigma$ with makespan $T$, the fringe jobs assigned to a machine of their native group and the core jobs assigned to the core group of their class form the set $\ijobs$ and we can set $\sigma' = \sigma|_\ijobs$.
The remaining jobs form the fractional jobs and they obviously fit fractionally into the left-over space, because we have a fitting integral assignment of them.
This also holds for the setups for groups with fractional jobs and no fringe jobs:
There has to be at least on setup for each such class on a machine that does not belong to their core group. 
Dropping the setups of the fringe jobs only increases the free space further. 

We consider the second claim.
Let $(\ijobs,\fjobs,\sigma')$ be a relaxed schedule with makespan $T$.
We construct a regular schedule and start by placing all the integral jobs like in the relaxed schedule.
To place the fractional jobs, we consider one speed group after another starting with group $0$.
For the current group $g$, we consider the jobs from $\fjobs'\subset \fjobs$, with $\fjobs' = \fjobs_{g-2}$, if $g > 0$, and $\fjobs' = \bigcup_{g'\leq -2} \fjobs_{g'}$, if $g = 0$.
We partition $\fjobs'$ into three sets $F_1,F_2,F_3$ that are treated differently.
The fringe jobs in $\fjobs'$ are included in the third group.
Let $k$ be a setup class. 
If the core jobs of $k$ have an overall size bigger than $s_k/\eps$, i.e., $\sum_{j\in\fjobs'\cap\corejobs_k} p_j > s_k/\eps$, they belong to $F_3$ as well.
Otherwise, they belong to $F_1$ if $k$ has a fringe job and to $F_2$, if it has none.

Let $k$ be a class whose fractional core jobs are included in $F_1$ or $F_2$. 
We will place the fractional core jobs of $k$ all on the same machine.
If the jobs are included in $F_1$, there exists a fringe job with class $k$ and we can place the fractional core jobs together with such a job.
A fringe job of class $k$ has a size of at least $s_k/\delta = s_k/\eps^2$, and hence the load due to the fringe job is increased at most by a factor of $(1+\eps)$ by this step.
This can happen at most once for each class and hence at most once for each fringe job.
Since all fringe jobs of the class could be fractional, we postpone this step until all the remaining fractional jobs are placed.
If, on the other hand, the fractional core jobs of class $k$ are included in $F_2$, we construct a container containing all respective jobs together with one setup of the class.
Note that the setup is already accounted for in the relaxed schedule, and that the overall size of the container is upper bounded by $(1 + 1/\eps)s_k$.
We call a container small on a machine $i$, if its size is upper bounded by $\eps v_i T$.
Each machine $i$ belonging to group $g$ or faster groups, is a fringe machines of class $k$ and therefore we have $s_k \leq \gamma v_i T$.
Hence, the size of the container is at most $(\eps^2 + \eps^3)v_i T \leq \eps v_i T$ (because $\eps \leq 1/2$), i.e., the container is small on $i$.
We place the container in the next step.

Next, we construct a sequence of jobs and containers and apply a greedy procedure to place them.
We start with an empty sequence and add all containers from the last step and all fringe jobs from $F_3$ in any order.
The core jobs from $F_3$ are added sorted by classes in the end of the sequence.
If there is a residual sequence that was not placed in the last iteration, we concatenate the two with the old sequence in the front. 
We now consider each of the machines $i\in M_g\setminus M_{g+1}$ with $L'_i < v_iT$ in turn and repeatedly remove the first job from the sequence and insert it on the current machine until the load of the machine exceeds $v_iT$.
Since all jobs and containers in the sequence are small on the machines of group $g$, they are overloaded at most by factor of $(1+\eps)$ afterwards.
For each step, the overall size of jobs and containers that are left in the sequence is at most the reduced accumulated fractional load $R_{g}$, because the remaining free space on the machines has either been filled completely, or the sequence is empty.
Since $R_{G} = W_{G} = W_{G-1} = 0$, all jobs and containers can be placed eventually.

Now, all jobs are properly placed, but some setups are still missing.
First, we consider core jobs that have been inserted in the greedy procedure and were not contained in a container. 
If the overall size of such core jobs of a class $k$ placed on a machine is at least $s_k/\eps$, adding the missing setups increases this size at most by a factor of $(1+\eps)$.
However, for each machine $i$, there can be at most two such classes $k$ without this property, namely the class that has been added first and the class that has been added last on the machine. 
For each class in between, all core jobs of this class in the sequence have been added to the machine, and these have sufficient overall size by construction.
Furthermore, if a job of class $k$ was placed on a machine $i$ in the greedy procedure, $i$ is a fringe machine of $k$.
Hence, the load of each machine $i$ after this step can be bounded by $(1+\eps)^2v_i T + 2\eps^3 v_iT\leq(1+\eps)^3v_i T$.
Lastly, we add the missing setups for the fringe jobs, resulting in an additional increase of at most $(1+\eps^2)$, because a fringe job of class $k$ has a size of at least $s_k/\eps^2$.
\end{proof}

\paragraph{Dynamic Program.}

To compute a relaxed schedule with makespan $T$ or correctly decide that there is none, we use a dynamic programming approach.
Therein, the groups of machine speeds are considered one after another starting with the slowest and going up.
For a fixed group the dynamic program can be seen as an adaptation of the one from \cite{europar16} for the identical case, and the overall structure of the program is similar to approaches used for the classical problem without setup times, e.g., in \cite{ptasUniform} and \cite{epstein2004approximation}.
However, there is some work to be done to combine these approaches and to deal with the fact, that the speed groups are overlapping.
In order to define the dynamic program and bound its running time, we first need some additional considerations and definitions.
For the sake of simplicity, we identify the set of classes $\mathcal{K}$ with the set of numbers $[K]$ in the following.

Let $B_g$ be the number of job sizes in $I$ that are big for at least one speed of group $g$.
We set $e(g) = \floor{\log \eps\check{v}_gT}$.
Because of the rounding of the job sizes, each size $p\in B_g$ is an integer multiple of $\eps 2^{e(g)}$.
Furthermore, we have $2^{e(g)} \leq \eps\check{v}_gT \leq p \leq \hat{v}_gT \leq \eps^{-1}\gamma^{-2}2^{e(g) + 1}$.
Hence, $|B_g| \leq 2/(\eps^2\gamma^2) = O(1/\eps^8)$.

We define a superset $L_g$ of possible load values that can occur on a machine that belongs to group $g$ and $g+1$ in a relaxed schedule due to integral jobs.
Such a machine may receive fringe jobs with native group $g$ or $g+1$, core jobs whose core group is one of these, as well as their setups.
The setup sizes have been rounded like the job sizes and for each of the mentioned setup sizes $s$ we have $s\geq \check{v}_{g}T $ and hence $s$ is an integer multiple of $\eps 2^{e(g)}$.
We set $L_g = \sett{ k \eps 2^{e(g)} }{ k \in \set{0,1,\dots,2/(\eps^2\gamma^3)} }$.
We have $|L_g| \leq 2/(\eps^2\gamma^3) + 1 = O(1/\eps^{11})$.

Next, we define a superset $\Lambda$ of possible load values of fractional jobs and corresponding setup sizes in a relaxed schedule.
Because of the first simplification step, each job and setup size is lower bounded by $\eps/(n+K)$ and $v_{\min} \geq \eps v_{\max} / m$.
We set $e^* = \floor{\log \eps/ (n+K)}$.
Because of the rounding, each job and setup size is a multiple of $\eps 2^{e^*}$.
Furthermore, the overall load of all jobs together with one setup of each class without a fringe job can be bounded by $mv_{\max}T \leq m^2/\eps$, or, more precisely, if this is not the case we can reject the current guess of the makespan.
Hence, we can set $\Lambda = \sett{k\eps 2^{e^*}}{k\in\set{0,1,\dots, 2m^2(n+K) / \eps^3}}$, and get $|\Lambda| = O(m^2(n+K) / \eps^3)$.
 
Lastly, we bound the number of speeds $|V_g|$ that occur in group $g$.
We have $\hat{v}_g = \check{v}_g / \gamma^2$ and applied geometric rounding on the speeds. 
Hence, $|V_g| = O(\log_{1+\eps} (1/\gamma^2) ) = O(1/\eps \log 1/\eps)$ (because $\eps < 1$).

A state of the dynamic program is of the form \[(g,k,\iota, \xi, \mu,\lambda)\]
with:
\begin{itemize}
	
	\item $g\in [G]$ is a group index.
	
	\item $k\in \set{0,\dots,K}$ is a setup class index including a dummy class $0$. The dummy class is included to deal with the fringe jobs with native group $g$.
	
	\item $\iota:B_g\rightarrow \set{0,\dots,n}$ is a function mapping job sizes to multiplicities. 
	Intuitively, $\iota(p)$ jobs of size $p$ corresponding to the current class still have to be dealt with in the current group.
	
	\item $\xi\in\set{0,1}$ is a flag that encodes whether a core job of the current class has been scheduled as a fractional job.
	
	\item $\mu: V_g\times L_{g-1}\cup L_g\times \set{0,1}\rightarrow \set{0,\dots,m}$ is a function mapping triples of machine speeds, load values and flags to machine multiplicities.
	We require, that $\mu(v,\ell,\zeta) = 0$, if $v\in V_g\cap V_{g+1}$ and $\ell\in L_g\setminus L_{g-1}$.
	Intuitively, we have $\mu(v,\ell,\zeta)$ machines of speed $v$ in the current machine group, with load $\ell$, that already received the setup of the current class ($\zeta = 1$) or not ($\zeta = 0$).
	
	\item $\lambda\in\Lambda^3$ is a load vector.
	Its values $\lambda_i$ corresponds to the load of fractional jobs together with the corresponding setups that have been pushed up to faster groups for the current ($i=1$), last ($i=2$), or some previous group ($i=3$) considered in the procedure.	
\end{itemize}

Let $\states$ be the set of states of the dynamic program.
Because of the above considerations, we have $|\states| = O(G K n^{\max_g |B_g|} m^{\max_g 2|V_g||L_{g-1}\cup L_g|}|\Lambda^3|) = (nmK)^{\mathrm{poly}(1/\eps)}$.
The states form the vertices of a graph, and the relaxed schedules correspond to paths from a start to an end state.
There are three types of edges:

\begin{enumerate}
	\item The edges marking the transition from a group $g$ to the next: 
	For each state $(g,k,\iota,\xi,\mu,\lambda)\in\states$ with $g<G$, $k=K$ and $\iota = 0$, there is an edge connecting the state with $(g+1,0,\iota', 0,\mu',\lambda')$, where $\iota'$, $\mu'$ and $\lambda'$ are defined as follows.
	For each $p\in B_{g+1}$ the value $\iota'(p)$ is the number of fringe jobs with native group $g$ and size $p$, i.e., $\iota'(p) = |\sett{j\in\fringejobs_g}{p_j = p}|$.
	We have $ \lambda'_1 = 0$, $\lambda'_2 = \lambda_1$, and:
	\[\lambda'_3 = \lambda_2 + \max\set[\Big]{0,\lambda_3 - \sum_{v\in V_{g}\cap V_{g-1}}\sum_{\ell\in L_{g-1}} (Tv - \ell) \cdot (\mu(v,\ell,0) + \mu(v,\ell,1)) }\]
	Furthermore, $\mu'(v,\ell,\zeta)$ is given by $\mu(v,\ell,0) + \mu(v,\ell,1)$, if $v\in V_g\cap V_{g+1}$, $\ell \in L_g$ and $\zeta = 0$; by $|\sett{i\in M_g}{v_i = v}|$, if $v\in V_{g+1}\setminus V_g$, $\ell = 0$ and $\zeta = 0$; and by $0$ otherwise.

	\item The edges marking the transition from one class to another: 
	For each state $(g,k,\iota,\xi,\mu,\lambda)\in\states$ with  $k < K$ and $\iota = 0$, there is an edge connecting the state with $(g,k+1,\iota',0,\mu',\lambda)$, where $\iota'$ and $\mu'$ are defined as follows.
	If $g$ is the core group of $k$, for each $p\in B_g$ the value $\iota'(p)$ is the number of core jobs of class $k$ and size $p$, i.e., $\iota'(p) = |\sett{j\in\corejobs_k}{p_j = p}|$, and otherwise $\iota' = 0$.
	Furthermore, we have $\mu'(v,\ell,0) = \mu(v,\ell,0) + \mu(v,\ell,1)$ and $\mu'(v,\ell,1) = 0$ for each $v\in V_g$ and $\ell\in L_{g-1}\cup L_g$. 
	
	\item The edges corresponding to scheduling decisions of the single jobs:
	For each $(g,k,\iota,\xi,\mu,\lambda)$ with $\iota\neq 0$ there are up to $2|V_g||L_{g-1}\cup L_g| + 1$ edges corresponding to the choices of scheduling some job on a machine with a certain speed and load, that already received a setup or not, or treating the job as fractional.
	Let $p\in B_g$ be the biggest size with $\iota(p) > 0$. 
	We define $\iota'$ as the function we get by decrementing $\iota(p)$.
	For each speed $v\in V_g$ and each load $\ell \in L_{g-1}\cup L_g$, we add up to two edges:
	If $\mu(v,\ell,0)>0$, $k>0$ and $\ell + p + s_k \leq vT$ we add an edge to the state $(g,k,\iota', \xi ,\mu',\lambda)$, where $\mu'$ is the function we get by decrementing $\mu(v,\ell,0)$ and incrementing $\mu(v, \ell + p + s_k, 1)$.
	If $\mu(v,\ell,0)>0$ and $\ell + p \leq vT$, we add an edge to the state $(g,k,\iota',\xi,\mu'',\lambda)$, where $\mu''$ is the function we get by decrementing $\mu(v,\ell,1)$ and incrementing $\mu(v, \ell + p, 1)$.
	Lastly, we add one edge to the state $(g,k,\iota',\xi',\mu,\lambda')$ with $\lambda'_2 = \lambda_2$ and $\lambda'_3 = \lambda_3$.
	If $k>0$, $k$ has no fringe job, and $\xi = 0$ we have $\xi' = 1$ and $\lambda'_1 = \lambda_1 + p + s_k$.
	Otherwise, $\xi' = \xi$ and $\lambda'_1 = \lambda_1 + p$.
\end{enumerate}
The start state of the dynamic program has the form $(0,0,\iota, 0, \mu,\lambda)$, with $\iota$, $\mu$, and $\lambda$ defined as follows.
For each $p\in B_{0}$ the value $\iota(p)$ is the number of fringe jobs with native group $0$ and size $p$; and for each speed $v\in V_0$, the value $\mu(v,0,0)$ is the number of machines with speed $0$.
Otherwise, we have $\mu(v,\ell,\zeta) = 0$.
For each $g\in\ZZ$, let $\mathcal{K}'_g\subseteq [K]$ be the set of classes with core group $g$ that do not have a fringe job.
We have $\lambda_1 = 0$, $\lambda_2 = \sum_{j\in\fringejobs_{-1}} p_j + \sum_{k\in\mathcal{K}'_{-1}}(s_k + \sum_{p\in\corejobs_k}p_j)$ and  $\lambda_3 = \sum_{g<-1}\big(\sum_{j\in\fringejobs_{g}} p_j + \sum_{k\in\mathcal{K}'_g}(s_k + \sum_{p\in\corejobs_k}p_j)\big)$.

The end states have the form $(G,K,0, 0, \mu',\lambda')$, where $\mu'$ and $\lambda'$ have the following form.
For each $v\in V_G$, we have $\mu'(v,\ell, \zeta) = 0$, if $\ell > vT$, and $\sum_{\ell\in L_{G-1}\cup L_G}\sum_{\zeta\in\{0,1\}}\mu'(v,\ell, \zeta) = |\sett{i\in M_G}{v_i = v}|$.
Furthermore, $\lambda'_1 = \lambda'_2 = 0$, and $\lambda'_3 \leq \sum_{v\in V_{G}}\sum_{\ell\in L_{G-1}} (Tv - \ell) \cdot (\mu(v,\ell,0) + \mu(v,\ell,1))$.

It can be easily verified that a relaxed schedule corresponds to a path from the start state to an end state, and that such a schedule can be recovered from such a path.
Hence, the dynamic program boils down to a reachability problem in a simple directed graph with $(nmK)^{\mathrm{poly}(1/\eps)}$ vertices.

\section{Unrelated Machines}
\label{sec:unrelated}
In this section, we study the problem of scheduling unrelated parallel machines with setup times.
Recall that for the classical model without setup times it is known \cite{lenstraUnrelated} that it cannot be approximated to within a factor of less than $\frac{3}{2}$ (unless P$=$NP) and that $2$-approximations are possible.
This is in stark contrast to our setting where, as we will see, the existence of classes and setups makes the problem significantly harder so that not even any constant approximation factor is achievable.
We approach the problem by formulating it as an integer linear program of which we round its optimal fractional solution by randomized rounding.
We will see in \cref{sec:unrelatedApprox} that this gives a tight approximation factor of $\Theta(\log n + \log m)$. 
In \cref{sec:unrelatedReduction}, we turn to inapproximability results and show that under certain complexity assumptions, this factor is essentially optimal.
We conclude with two special cases that admit constant factor approximations in \cref{sec:constantCases}.

Consider the following integer linear program \texttt{ILP-UM}, describing the problem at hand:
	\begin{align}
	\sum_{j \in \mathcal{J}} x_{ij}p_{ij} + \sum_{k \in \mathcal{K}} y_{ik} s_{ik} & \leq T & \forall i \in \mathcal{M} \label{in:makespan}\\
	\sum_{i \in \mathcal{M}} x_{ij} &=1 & \forall j \in \mathcal{J} \label{in:fullyAssigned}\\
	x_{ij}, y_{ik} &\in \{0,1\} &\forall i \in \mathcal{M},j \in \mathcal{J}, k \in \mathcal{K} \label{in:integral}\\
        y_{ik_j} &\geq x_{ij} &  \forall i \in \mathcal{M}, j \in \mathcal{J} \label{in:setup}\\
        x_{ij} &= 0 & \forall i \in \mathcal{M}, j \in \mathcal{J} : p_{ij} > T \label{in:tooLarge}
	\end{align}
For each job $j$, there is an assignment variable $x_{ij}$ stating whether or not job $j$ is assigned to machine $i$.
Additionally, for each class $k$ there is one variable $y_{ik}$ indicating whether or not machine $i$ has a setup for class $k$. 
Then, \cref{in:makespan} ensures that the load, given by processed jobs and setups, on each machine does not violate the desired target makespan $T$.
\cref{in:fullyAssigned,in:integral} make sure that each job is completely assigned to one machine.
By \cref{in:setup} it is guaranteed that if a job $j$ of class $k_j$ is assigned to machine $i$, then a setup for class $k_j$ is present on machine $i$.
\cref{in:tooLarge} guarantees that no job $j$ that is too large on machine $i$ to be finished within the desired makespan bound is assigned to machine $i$.

\subsection{Approximation Algorithm}
\label{sec:unrelatedApprox}
Starting with an optimal solution $(x^*, y^*)$ to the linear relaxation of \texttt{ILP-UM} where we replace \cref{in:integral} by $0\leq x_{ij}, y_{ik} \leq 1$, we can use the following approach to compute an integral solution approximating an optimal schedule:

\begin{enumerate}
    \item For each $i \in \mathcal{M}$ and $k \in \mathcal{K}$, set $y_{ik} = 1$ with probability $y_{ik}^*$ (perform a setup for $k$ on $i$) and
    $y_{ik} = 0$  with probability $1-y_{ik}^*$.\\
    If $y_{ik} = 1$, then, for each job $j$ with $k_j = k$, set $x_{ij} =1$ (assign $j$ to $i$) with probability $x^*_{ij}/y^*_{ik}$ and $x_{ij} =0$ with probability $1- (x^*_{ij}/y^*_{ik})$.\label[step]{step:open}
    \item Repeat \cref{step:open} $c \log n$ times.
    \item If there are unassigned jobs left, then schedule each job $j \in \mathcal{J}$ on machine $\text{argmin}_{i \in \mathcal{M}}\{p_{ij}\}$.
    \item If a job is assigned to multiple machines, remove it from all but one.
    If a class's setup occurs multiple times on a machine, remove all but one.
\end{enumerate}
The following analysis already appeared in a fairly similar way in \cite{khuller10}.
However, for the sake of completeness and due to small adaptations, we restate it in the following.

\begin{lemma}
Step 5.\ is executed with probability at most $1/n^c$.
\end{lemma}

\begin{proof}
Consider a fixed job $j \in \mathcal{J}$ and a fixed iteration $h, 1 \leq h \leq c \log n$.
Let $\mathcal{\bar A}_{ij}^h$ be the event that job $j$ is not assigned to machine $i$ after iteration $h$.
Let $\mathcal{\bar A}_{j}^h$ be the event that job $j$ is not assigned to any machine after iteration $h$.
We have 
\begin{equation}
\Pr[\mathcal{\bar A}_{ij}^h| \mathcal{\bar A}_{j}^{h-1}] = (1-y^*_{ik_j}) + y^*_{ik_j}\left(1-\frac{x^*_{ij}}{y^*_{ik_j}}\right) = 1-x^*_{ij} \label{in:probNotAssigned} .
\end{equation}
Taking into account all $m$ machines, we then have
\begin{equation}
\Pr[\mathcal{\bar A}_{j}^h| \mathcal{\bar A}_{j}^{h-1}] \overset{(\ref{in:probNotAssigned})}{\leq} \prod_{i \in \mathcal{M}} (1-x^*_{ij}) \overset{(\ref{in:fullyAssigned})}{\leq} \left(1-\frac{1}{m}\right)^m \leq \frac{1}{e}. \label{in:probNotAssigned2}
\end{equation}
Hence, for the probability that $j$ is not assigned to any machine after $h$ iterations we have
\begin{align*}
\Pr[\mathcal{\bar A}_{j}^h] &= \Pr[\mathcal{\bar A}_{j}^h | \mathcal{\bar A}_{j}^{h-1}] \cdot \Pr[\mathcal{\bar A}_{j}^{h-1}]\\ 
& = \ldots\\
& = \Pr[\mathcal{\bar A}_{j}^h | \mathcal{\bar A}_{j}^{h-1}] \cdot \Pr[\mathcal{\bar A}_{j}^{h-1} | \mathcal{\bar A}_{j}^{h-2}] \cdot \ldots \cdot \Pr[\mathcal{\bar A}_{j}^1]
\overset{(\ref{in:probNotAssigned2})}{\leq} \left(\frac{1}{e}\right)^{h},
\end{align*}
and hence for $h = c \log n$, we obtain the lemma.
\end{proof}

In the next lemma we show that the expected load assigned to a machine per iteration is bounded by $O(T)$.
This together with the previous lemma, then proves the final result.
Compared to \cite{khuller10}, there is a slight difference in our proof:
If $q_{ij}$ describes the probability that job $j$ is assigned to machine $i$ in an iteration of the randomized rounding algorithm, then in \cite{khuller10} the authors can (and do) use the fact that $\sum q_{ij} p_{ij} \leq T$.
This, however, is not true in our case due to different constraints in the underlying linear program.

\begin{lemma}
Let $L_i$ describe the load on machine $i$ after the $c \log n$ iterations.
Then, $\Pr[L_i = O(T (\log n + \log m)) \enspace \forall i \in \mathcal{M}] = 1-1/n^c$.
\end{lemma}

\begin{proof}
Let us first consider the load on the machines due to processed jobs.
Let $Z_{ij}^h$ be a random variable with 
\[Z_{ij}^h =
\begin{cases}
p_{ij}/T, & \text{if } j \text{ assigned to } i \text{ in iteration } h\\
0, & \text{otherwise.}
\end{cases}\]\\
Let $Z_i^\mathcal{J} = \sum_{h = 1}^{c \log n}\sum_{j \in \mathcal{J}} Z_{ij}^h$.
Then, we have 
\[
\mathbb{E}[Z_i^\mathcal{J}] = \frac{1}{T} \sum_{h = 1}^{c \log n} \sum_{k \in \mathcal{K}} \sum_{j \in \mathcal{J}: k_j=k} 0 \cdot (1-y^*_{ik}) + y^*_{ik}\left(\sum_{j: k_j = k} \frac{x^*_{ij}}{y_{ik}^*}p_{ij}\right) \overset{(\ref{in:makespan})}{\leq} c \log n.
\]
Using the essentially same reasoning to analyze the load on the machines due to setups and denoting $Z_i^\mathcal{S}$ the analog of $Z_i^\mathcal{J}$, we also have $\mathbb{E}[Z_i^\mathcal{S}] \leq c \log n$.
Because all $Z_i$ are sums of independent random variables with values in $[0,1]$, we can now apply standard Chernoff-bounds and obtain for $\delta \coloneqq 3(\frac{\log(n+m)}{c\log n}+1)$ that 
$\Pr[\exists i: L_i \geq (1+\delta)Tc\log n] \leq \Pr[\exists i, x \in \{\mathcal{S},\mathcal{J}\}: Z^x_i \geq (1+\delta)c\log n] \leq (m+n) \exp(-\frac{1}{3}\delta c \log n )\leq (1/n)^c$.
\end{proof}

Taking the last two lemmas together with the fact that the makespan is always upper bounded by $O(T \cdot n)$, we obtain the following theorem.
\begin{theorem}
With high probability and on expectation the randomized rounding approach provides a solution with makespan $O(T(\log n+\log m))$.
\end{theorem}

By choosing the parameter $c$ sufficiently large when applying the algorithm within the dual approximation framework, we obtain an approximation factor of $O(\log n + \log m)$.
Also, it is not too hard to see that this bound is actually tight as one can prove an integrality gap of $\Omega(\log n +\log m)$ for the linear relaxation of \texttt{ILP-UM}.
This can be shown by using a construction following the ideas for proving the integrality gap for set cover (e.g.\ \cite[p.~111-112]{approxBook}).

\begin{corollary}
There is a polynomial time randomized algorithm with approximation factor $O(\log n + \log m)$, which matches the integrality gap of the linear relaxation of \texttt{ILP-UM}.
\end{corollary}

\subsection{Hardness of Approximation}
\label{sec:unrelatedReduction}
We now show that the approximation factor of $\Theta(\log n + \log m)$ is (asymptotically) optimal unless all problems in NP have polynomial-time Monte Carlo algorithms.
Recall that the complexity class RP (Randomized Polynomial-Time) is defined as the class of problems $L$ for which there is a randomized algorithm running in polynomial time and with the following properties (e.g.\ see \cite{zppDef}):
\begin{itemize}
    \item If the input $x \notin L$, the algorithm outputs ``\textsc{No}'' with probability 1.
    \item If the input $x \in L$, the algorithm outputs ``\textsc{Yes}'' with probability at least $1/2$.
\end{itemize}
Therefore, if such an algorithm outputs ``\textsc{Yes}'', it provides the correct answer; if it, however, outputs ``\textsc{No}'', it might err.

In what follows, we show the following result on the hardness of approximating our problem on unrelated machines.
\begin{theorem}
\label{th:inapprox}
Scheduling with setup times on unrelated machines cannot be approximated within a factor of $o(\log n + \log m)$ in polynomial time unless $\text{NP} \subset \text{RP}$.
This even holds for the restricted assignment case.
\end{theorem}
To do so, we reduce from the following formulation of the well-known \textsc{SetCover} problem: 
In \textsc{SetCoverGap} there is given a universe $\mathcal{U}$ of $N \coloneqq |\mathcal{U}|$ elements and a collection of $m$ subsets of $\mathcal{U}$.
The goal is to decide whether there is a solution covering $\mathcal{U}$ that consists of $t$ subsets or if (at least) $\alpha t$ subsets are needed.
We call an instance with the former property a \textsc{Yes}-instance and with the latter a \textsc{No}-instance.
A result from \cite{setCoverNpHardness} shows the following lemma. 
\begin{lemma}[Theorem 7 in \cite{setCoverNpHardness}]
\label{le:setCoverHard}
There exists a $t$ such that it is NP-hard to decide \textsc{SetCoverGap} for $\alpha = \Theta(\log N)$ and $\log m = O(\log N)$.
\end{lemma}

The idea of our reduction is to exploit the apparent connection between \textsc{SetCover} and our unrelated machines variant:
Each set is mapped to a machine and each element is mapped to a job. 
A machine can process a job if and only if the respective set contains the respective element.
Additionally assuming that all jobs belong to the same class, by this we see that a \textsc{Yes}-instance requires much less setups than a \textsc{No}-instance.
Unfortunately, this not yet leads to a respectively high and small makespan.
However, by creating a larger number of classes and randomizing the mapping between sets and machines, we can achieve a (more or less) even distribution of setups that need to be done and hence, depending on the type of the \textsc{SetCoverGap} instance, a high or small makespan.
We formalize this idea in the proof of \cref{th:inapprox}.
\begin{proof}
Given an instance $I$ for \textsc{SetCoverGap}, we construct an instance $I'$ for our problem with the following properties:
\begin{enumerate}
    \item The reduction can be done in polynomial time and $I'$ consists of $n = \Theta(N^c)$ jobs, for some constant $c$.
    \item If $I$ is a \textsc{No}-instance, then $I'$ has a makespan of at least $\Omega(\frac{K}{m} \cdot \alpha t)$.
    \item If $I$ is a \textsc{Yes}-instance, then $I'$ has a makespan of at most $O(\frac{K}{m} \cdot t)$ with probability at least $1/2$.
\end{enumerate}
Consequently, there is a gap of $\Omega(\alpha)$ and by Property~1.\ and \cref{le:setCoverHard}, $\alpha = \Omega(\log n)$ and $\alpha = \Omega(\log m)$ and the existence of a polynomial-time algorithm with approximation factor $o(\log n + \log m)$ for our problem makes the problem $\textsc{SetCoverGap}$ solvable in expected polynomial time, yielding the theorem.

We now show how to construct $I'$.
In instance $I'$ there are $m$ unrelated machines and $K = \frac{m}{t} \log m$ classes.
All setup times are set to be $1$, that is, $s_{ik} = 1$ for all $i \in \mathcal{M}, k \in \mathcal{K}$.
The jobs $\{j_1^k, j_2^k, \ldots, j_N^k\}$ of class $k=1,2,\ldots,K$ are defined by the $N$ elements in $I$ in the following way:
We choose a permutation $\pi_k : \mathcal{M} \to \mathcal{M}$ at random (and independent from the choices of $\pi_{k'}$ for $k' \neq k$).
Then, for each element $e$ in the \textsc{SetCoverGap} instance $I$, we create a job $j_e^k$ in instance $I'$ that has a size $p_{ij_e^k} = 0$ if $e \in S_{\pi_k(i)}$ and $p_{ij_e^k} = \infty$ otherwise.

Next, we take a look at the makespan of $I'$ if $I$ is a \textsc{No}-instance.
In this case, at least $\alpha t$ sets are needed to cover all elements. 
However, this implies that for each class at least that many machines are needed to process all jobs (or otherwise the makespan is $\infty$).
Therefore, by summing over all $K$ classes, at least $K \cdot \alpha t$ setups need to be performed.
By an averaging argument this leads to the existence of a machine with makespan of at least $\frac{K}{m} \cdot \alpha t$.

We now turn our attention to the case where $I$ is a \textsc{Yes}-instance and show that with probability at least $1/2$ there is a solution with makespan $O(\frac{K}{m} \cdot t)$.
To this end, we setup a machine $i$ for class $k$ (and process all jobs $j$ of class $k$ on machine $i$ that fulfill $p_{ij}=0$) if $S_{\pi_k(i)}$ is part of the solution to $I$.
Therefore, each class is setup on $t$ of the $m$ machines.
For a fixed machine $i$ and a fixed class $k$, the probability that $i$ is setup for $k$ is consequently $t/m$ since $\pi_k(i)$ is chosen uniformly at random.
Also, the probability that $i$ is setup for all classes of a fixed subset of $r$ classes is $(t/m)^r$ as the $\pi_k$ are chosen independently.
Therefore, the probability that a fixed machine $i$ is setup for at least $r$ classes is upper bounded by
\[
\binom{K}{r} \left(\frac{t}{m}\right)^r \leq \left(\frac{Ket}{rm}\right)^r \enspace .
\]
Hence, for the probability that there is some machine which is setup for at least $r\coloneqq 2Ket/m + 2\log m = O(\frac{K}{m} \cdot t)$ classes is (for $m \geq 2$) upper bounded by 
\[
m \cdot \left(\frac{Ket}{rm}\right)^r \leq m \cdot \left(\frac{1}{2}\right)^{2\log m} \leq \frac{1}{m} \leq \frac{1}{2} \enspace .
\]
Therefore, $I'$ has a makespan of at most $O(\frac{K}{m} \cdot t)$ with probability at least $1/2$.

Also note that $\log (n) = \log (K \cdot N) \leq \log(m \log m \cdot N) = O(\log N)$, where the last equality holds due to the polynomial relation between $m$ and $N$ according to \cref{le:setCoverHard}.
This concludes the proof.
\end{proof}

\subsection{Special Cases with Constant Approximations}
\label{sec:constantCases}
In this section, we identify and approximate two special cases of unrelated machines, for which constant factor approximations are possible. 
Both cases require classes to have certain structural properties that make the reduction and hence, the inapproximability from \cref{sec:unrelatedReduction} invalid:
Either we consider the restricted assignment case with the additional assumption that the set of eligible machines is the same for all jobs of a class, or we assume that, on each machine, all jobs of a given class have the same processing times.

\subsubsection{Restricted Assignment with Class-uniform Restrictions}
\label{sec:unrelatedRestrictedAssignment}
Although even the restricted assignment variant of our scheduling problem cannot be approximated with a factor of $o(\log n)$ as shown in \cref{th:inapprox}, in this section we will see that the following special case admits a much better approximation factor. 
Let the \emph{restricted assignment problem with class-uniform restrictions} be defined as the restricted assignment problem with the additional constraint that for all $j, j' \in \mathcal{J}$ with $k_j = k_{j'}$ it holds $M_j = M_{j'}$.
That is, all jobs of a class $k$ have the same set of eligible machines and by abuse of notation call this set $M_k$.

Note that we can add the following valid constraints given by \cref{in:classMakespan,in:infeasibleJob,in:infeasibleClass} to \texttt{ILP-UM}: 
\begin{align}
        \sum_{j : k_j = k} x_{ij} p_{ij} + y_{ik} s_{ik} &\leq y_{ik} T & \forall i \in \mathcal{M}, \forall k \in \mathcal{K} \label{in:classMakespan}\\
	x_{ij}&=0 & \forall i \in \mathcal{M}, j \in \mathcal{J} : p_{ij} + s_{ik_j}> T \label{in:infeasibleJob}\\
	y_{ik}&=0 & \forall i \in \mathcal{M}, k \in \mathcal{K} : s_{ik} > T \label{in:infeasibleClass}
\end{align}
\cref{in:classMakespan} holds because a job of a class $k$ can only by processed on a machine $i$ if this machine is setup for class $k$.
Additionally, \cref{in:infeasibleJob,in:infeasibleClass} avoid the assignment of jobs to machines where the setup or the job's processing time is too large.
Let \texttt{ILP-RA} denote the program given by \cref{in:makespan,in:fullyAssigned,in:integral,in:setup,in:classMakespan,in:infeasibleJob,in:infeasibleClass}.
Unfortunately, we do not know how to round a solution to the linear relaxation of \texttt{ILP-RA} to a good approximation for our problem.
However, instead we formulate a different, relaxed linear program \texttt{LP-RelaxedRA}, which we will utilize for our approximation algorithm:
	\begin{align}
	\sum_{k \in \mathcal{K}} \bar x_{ik}(\bar p_{ik} + \alpha_{ik} s_{ik}) & \leq T & \forall i \in \mathcal{M} \label{in:splittableMakespan}\\
	\sum_{i \in \mathcal{M}} \bar x_{ik} &= 1 & \forall k \in \mathcal{K}\label{in:wholeClass}\\
	\bar x_{ik} & \geq 0 &\forall i \in \mathcal{M},k \in \mathcal{K}\label{in:nonNegative}\\
	\bar x_{ik}&=0 & \forall i \in \mathcal{M}, k \in \mathcal{K}: s_{ik}> T \label{in:infeasibleClass2}
	\end{align}
This linear program takes a different view in the sense that it does not operate on the level of jobs but instead it has a variable $\bar x_{ik}$ for each class-machine-pair determining the fraction of (the workload of) class $k$ processed on machine $i$.
Therefore, let $\bar p_{ik} \coloneqq \sum_{j: k_j = k} p_{ij}$ be the overall workload of class $k$ if its jobs can be processed on machine $i$, and $\bar p_{ik} = \infty$ otherwise. 
Also, let $\alpha_{ik} \coloneqq \max\left\{1, \frac{\bar p_{ik}}{T-s_{ik}}\right\}$.
If $x$ is a feasible solution to \texttt{ILP-RA}, 
then $\bar x$ with $\bar x_{ik} \coloneqq \sum_{j: k_j=k} x_{ij}\frac{p_{ij}}{\bar p_{ik}}$ is a feasible solution to \texttt{LP-RelaxedRA} as the next lemma proves.
\begin{lemma}
Let $x$ be a feasible solution to \texttt{ILP-RA}.
Then $\bar x$ is a feasible solution to \texttt{LP-RelaxedRA}.
\end{lemma}

\begin{proof}
First, note that \cref{in:infeasibleClass2} directly follows from \cref{in:infeasibleJob,in:infeasibleClass}.
\cref{in:wholeClass} is satisfied as we have 
\begin{align*}
\sum_{i \in \mathcal{M}} \bar x_{ik}  = \sum_{i \in \mathcal{M}} \sum_{j: k_j=k} x_{ij}\frac{p_{ij}}{\bar p_{ik}}
 &= \sum_{i \in M_k}\frac{1}{\bar p_{ik}} \sum_{j: k_j=k} x_{ij}p_{ij}\\
 &=\frac{1}{\bar p_{k}}  \sum_{j: k_j=k} p_{j} \sum_{i \in M_k}  x_{ij} = 1,
\end{align*}
where the first equality follows by definition of $\bar x_{ik}$, the second because $\bar x_{ik} = 0$ if $i \notin M_k$ and the last one because $\sum_{i \in M_k} x_{ij} = 1$ by \cref{in:fullyAssigned}.

To see why \cref{in:splittableMakespan} holds, first observe that for $x$ we have 
\begin{equation}
\sum_{j \in \mathcal{J}} x_{ij}p_{ij} + \sum_{k \in \mathcal{K}} \max\left\{\max_{j: k_j = k} x_{ij}, \frac{\sum_{j:k_j=k}x_{ij} p_{ij}}{T-s_{ik}}\right\} s_{ik}  \leq T \ \forall i \in \mathcal{M} \label{in:alphaStyle}
\end{equation}
due to \cref{in:makespan,in:setup,in:classMakespan}.
Then we have 
\begin{align*}
\sum_{k \in \mathcal{K}} \bar x_{ik}(\bar p_{ik} + \alpha_{ik} s_{ik}) &= \sum_{k \in \mathcal{K}} \sum_{j: k_j=k} x_{ij}\frac{p_{ij}}{\bar p_{ik}} \sum_{j: k_j =k} p_{ij} + \sum_{k \in \mathcal{K}} \sum_{j: k_j=k} x_{ij}\frac{p_{ij}}{\bar p_{ik}} \alpha_{ik}s_{ik}\\
& = \sum_{j \in \mathcal{J}} x_{ij} p_{ij} + \sum_{k \in \mathcal{K}} \frac{\alpha_{ik}}{\bar p_{ik}} \sum_{j: k_j=k} x_{ij} p_{ij} s_{ik} \leq T,
\end{align*}
where the first and second equality follow from definition of $\bar x_{ik}$ and $\bar p_{ik}$ respectively, and the last inequality holds due to \cref{in:alphaStyle}.
\end{proof}

\texttt{LP-RelaxedRA} is identical to the LP given in \cite{correa15}.
There it is shown that an extreme solution to the LP can be rounded to a solution with makespan at most $2T$ that is feasible for the problem where jobs can be split arbitrarily but each part requires a (job-dependent) setup. 
Interestingly, even though in our model setups are associated with classes and even more crucial, we do not allow jobs to be split, the (essentially) same approach they use, provides an approximation factor of $2$ for our problem, too.
The high-level idea how to obtain a $2$-approximation based on an optimal (extreme) solution for \texttt{LP-RelaxedRA} is as follows:
It is known that due to the structure of \texttt{LP-RelaxedRA}, the graph representing the solution is a pseudo-forest.
We can exploit this fact to modify the solution such that it has a makespan of at most $2T$, but in which (additionally) each machine processes at most one class partly (but not completely) and in which, for each class $k$, the property holds that from the set of machines processing parts of $k$ at most one machine has a load larger than $T$.
This allows us to greedily assign the actual jobs according to the (modified) fractional solution to the machines and thereby increasing the load per machine (with load at most $T$) by at most one setup plus one job of the same class and hence, by at most $T$.
The details are given next and for the sake of completeness, we restate the rounding procedure together with its properties from \cite{correa15}:
Given an extreme solution $\bar x^*$ to \texttt{LP-RelaxedRA}, all variables $\bar x_{ik}$ with $\bar x^*_{ik} \in \{0,1\}$ will remain unchanged, are excluded from our further considerations and class $k$ is processed on machine $i$ if $\bar x^*_{ik} = 1$.
Let $G = (V,E)$ be the bipartite graph on node set $V = \mathcal{K} \setminus \{k : \exists i \text{ with } \bar x^*_{ik} =1\} \cup \mathcal{M}$ and edge set $E = \{\{i,k\} : 0<\bar x^*_{ik}<1\}$.
$G$ forms a graph in which each connected component is a pseudotree. 
For the sake of rounding, we now construct a subset $\tilde E \subseteq E$ of edges as follows: 
For each connected component, let $C$ be the unique cycle (or an arbitrary path if no cycle exists) and let $J(C)$ be the nodes in $C$ corresponding to classes.
Fix an arbitrary direction along $C$ and starting at an arbitrary node $v \in J(C)$, remove each second edge along $C$ starting with the edge leaving $v$.
We then end up with a graph only consisting of trees.
In the last step, for each class $k \in J(C)$ belonging to $C$ we build a directed tree rooted in $k$ by directing edges away from the root.
Then we remove all edges leaving machine nodes. 
All edges that remain after these two steps belong to $\tilde E$.

It is not too hard to see, and it is formally proven in \cite{correa15}, that we have the following lemma.
\begin{lemma}[\cite{correa15}]
\label{le:correaLemma}
By the construction described above, we have the following two properties for a schedule induced by $\tilde E$:
\begin{enumerate}
    \item Each machine $i$ processes at most one class $k$ with $\{i,k\} \in \tilde E$, and 
    \item for each class $k$ there is at most one machine $i$ such that $\{i,k\} \notin \tilde E$ (and $\bar x^*_{ik} > 0$).
\end{enumerate}
\end{lemma}

For each class $k$ we choose an arbitrary machine $i^+_k$ such that $\{i^+_k,k\} \in \tilde E$.
If there is a machine $i^-_k$ such that $\bar x^*_{i^-_kk}>0$ but $\{i^-_k,k\} \notin \tilde E$, we move all workload of $k$ processed on $i^-_k$ from $i^-_k$ to $i^+_k$ and add a (full) setup for $k$ to machine $i_k^+$.
By this and \cref{le:correaLemma} we then have the property that each machine processes at most one class fractionally.
Let $M(k)$ be the set of machines that process (parts of) class $k$.
We next prove the following lemma.

\begin{lemma}
For any $k \in \mathcal{K}$, all machines in $M(k) \setminus \{i^+_k\}$ have a load of at most $T$.
The load of $i^+_k$ is upper bounded by $2T$.
\end{lemma}

\begin{proof}
Consider a class $k$.
Note that $i^+_k \neq i^+_{k'}$ for all $k'\neq k$ by \cref{le:correaLemma}.
Hence, the first statement of the lemma holds.
The second statement follows by an observation already made in \cite{correa15}: 
By the constraints of \texttt{LP-RelaxedRA}, $\alpha_{i_k^-k} \bar x^*_{i_k^-k} \leq 1$ and by definition of $\alpha_{i_k^-k}$, we have $T \geq \bar p_{i_k^-k}/\alpha_{i_k^-k} + s_{i_k^-k}$.
Taken together, $\bar x^*_{i_k^-k} \bar p_{i_k^-k} + s_{i_k^-k} \leq T$ and because we consider restricted assignment with class-uniform restrictions, we also have
$\bar x^*_{i_k^-k} \bar p_{i_k^+k} + s_{i_k^+k} \leq T$, proving the lemma.
\end{proof}

Finally, we need to explain how to obtain the final feasible schedule with makespan at most $2T$. 
Obtaining a feasible schedule from the solution so far, requires adding a (full) setup for class $k$ on all machines $i \in M(k) \setminus \{i^+_k\}$ as well as showing how to actually assign the jobs of $k$ to the machines $i \in M(k)$.
We say that a time slot of size $x$ is reserved for class $k$ on a machine $i$ if $\bar x^*_{ik}\bar p_{ik} = x$. 
For any fixed class $k$, sort the machines in $M(k)$ so that machine $i_k^+$ comes last in this ordering.
Starting with the first machine in the ordering, take the jobs of $k$ and greedily fill them into the reserved time slots by assigning the current job to the current machine if the reserved time slot is not yet full. 
As soon as a machine is full, proceed with the next machine.
It is not hard to see that by this procedure the load of each machine $i \in M(k) \setminus \{i^+_k\}$ is increased by an additive of at most $s_{ik} + \max_{j : k_j =k} p_{ij} \leq T$ and the last machine $i_k^+$ keeps its load of at most $2T$.
Therefore, we have proven the desired result.

\begin{theorem}
\label{th:raResult}
The restricted assignment problem with class-uniform restrictions admits a $2$-approximation.
\end{theorem}

\subsubsection{Unrelated Machines with Class-uniform Processing Times}
\label{sec:specialUnrelated}
A second special case that allows constant factor approximations is the one of unrelated machines in which all jobs of a given class have the same processing times on any machine.
That is, for all $i\in \mathcal{M}$ and $j,j' \in \mathcal{J}$ it holds $k_j = k_{j'}$ implies $p_{ij} = p_{ij'}$.

We solve this problem similarly to the restricted assignment problem with class-uniform restrictions in the previous section.
To do so, we modify the approach as follows:
First of all, we replace \cref{in:infeasibleClass2} in \texttt{LP-RelaxedRA} by 
    \begin{equation}
    \bar x_{ik} = 0 \enspace \forall i \in \mathcal{M},k \in \mathcal{K}: s_{ik} + p_{ij} > T \text{ for some } j \text{ with } k_j =k \enspace . \label{in:specialUnrelatedSizeBound}
    \end{equation}
(Note that this is a valid constraint since all jobs of a class $k$ have the same size on machine $i$ and if a job together with its class' setup does not fit to a machine, no workload of class $k$ will be assigned to $i$ at all.)
Then we construct the set $\tilde E$ with the properties of \cref{le:correaLemma} as before.
Now, for each class $k \in \mathcal{K}$ let $i^-_k$ be the machine such that $\bar x_{{i^-_k}k}^* > 0$ but $\{i^-_k,k\} \notin \tilde E$ (if it exists).
Let $i^+_{k,\iota}$, $\iota = 1, \ldots, \iota_k$ be the machines such that $\bar x_{{i^+_{k,\iota}}k}^* > 0$ and $\{i^+_{k,\iota},k\} \in \tilde E$.
In case $\bar x_{{i^-_k}k}^* > \frac{1}{2}$, process the entire class $k$ on machine $i^-_k$.
Otherwise, distribute the amount of $k$ processed on $i^-_k$ proportionally to the machines $i^+_{k,\iota}$.
That is, set $\bar x_{{i^-_k}k}^* = 0$ and $\bar x_{{i^+_{k,\iota}}k}^* = 2 \bar x_{{i^+_{k,\iota}}k}^*$.
After these steps, the load on each machine is at most $2T$. 
Finally, it remains to add at most one setup to each machine and, as before, to greedily fill the reserved slots by the actual jobs.
This increases the load on each machine by an additive of at most $T$ due to \cref{in:specialUnrelatedSizeBound} and hence, we have constructed a $3$-approximation.
Together with a straightforward adaptation of the reduction given in \cite{correa15}, we have the following result.
    
\begin{theorem}
The unrelated machines case with class-uniform processing times admits a $3$-approximation.
It cannot be approximated to within a factor less than $2$ unless $P = NP$.
\end{theorem}

\bibliography{references}
\bibliographystyle{plain}
\end{document}